\documentclass[11pt,a4paper]{article}
\def\iindex{}
\usepackage[centertags]{amsmath}
\usepackage{amsfonts}
\usepackage{amssymb}
\usepackage{amsthm}
\usepackage{epsfig}
\usepackage{setspace}
\usepackage{ae} 
\usepackage{eucal} 
\usepackage[authoryear]{natbib}

\usepackage[usenames]{color}

\theoremstyle{plain}
\newtheorem{thm}{Theorem}[section]
\newtheorem{lem}[thm]{Lemma}
\newtheorem{cor}[thm]{Corollary}
\newtheorem{prop}[thm]{Proposition}
\theoremstyle{definition}
\newtheorem{defi}[thm]{Definition}

\theoremstyle{remark}
\newtheorem{exmp}[thm]{Example}

 \DeclareMathOperator*{\esssup}{esssup}
\DeclareMathOperator*{\essinf}{essinf}
\begin{document}

\title{A First-Order BSPDE for Swing Option Pricing\index{Existence}}

\author{Christian Bender$^1$, Nikolai Dokuchaev$^2$}

 \maketitle
\footnotetext[1]{Saarland University, Department of Mathematics,
Postfach 151150, D-66041 Saarbr\"ucken, Germany, {\tt
bender@math.uni-sb.de}.} \footnotetext[2]{       Department of
Mathematics \& Statistics, Curtin University, GPO Box U1987, Perth,
6845 Western Australia, Australia, {\tt N.Dokuchaev@curtin.edu.au }}

\begin{abstract}
We study an optimal control problem related to swing option pricing in a general non-Markovian  setting in continuous time. As a main result we
show that the value process solves a first-order non-linear backward stochastic partial differential equation. Based on this result we can characterize the set of optimal controls and
derive a dual minimization problem.
\\[0.2cm] \emph{Keywords:} Backward SPDE, stochastic optimal control, swing options.
\\[0.2cm] \emph{AMS classification:} 60H15; 49L20; 91G20.
\end{abstract}

\section{Introduction}

In a swing option contract, the holder of the option can buy some volume of a commodity, say electricity, for a fixed strike price during the lifetime of the option. There are typically local constraints on how much volume can be
exercised at a given time, and global constraints on the total volume. Swing options are particularly popular in electricity markets, and can be used to hedge against
the risk of  fluctuating demand, see \citet{CL}.

Mathematically, the pricing problem of such a swing option leads to
optimal control problems, whose formulation varies depending on the
way the constraints are formulated. On the one hand, the constraints
can be formulated discretely in the following sense: The total
volume must be exercised in form of a finite number of packages.
Local constraints prescribe how many packages can at most be
exercised at a given time and refraction periods are imposed to
enforce a minimal waiting time after one package is exercised. This
formulation leads to multiple stopping problems and was studied in
discrete time e.g. by \citet{JRT}, \citet{MH}, \citet{Be1},
\citet{Sch}, and \citet{BSZ} and in continuous time by \citet{CT},
\citet{ZM}, and \citet{Be2}. On the other hand, constraints can be
imposed on the rate at which the option is exercised. This approach
leads to a continuous time optimal control problem as stated by
\citet{Ke04} in a general framework and studied by \citet{BLM} in a
diffusion setting; see also the related work of \citet{Do2013} for a
more general notion of controlled options. Related discrete time
optimal control formulations for swing option pricing can be found
e.g. in \citet{Bal}, \citet{Ba}.

In the present paper we adopt the second approach and formulate the local constraint in continuous time in terms of the rate of exercising. Suppose an adapted process $X(t)$ denotes the discounted payoff of
the option, if one unit volume is exercised at time $t$. In the case of swing option pricing we can set $X(t)=e^{-\rho t}(S(t)-K)_+$, where $S$ is the electricity price process, $K$ is the strike price,
$\rho$ is the interest rate, and $(x)_+$ stands for the positive part of $x$. Then, we consider the following  control problem
$$
\bar J(t, y):=\esssup_{u} E\left[\left.\int_t^T u(s)X(s)ds
\right|\mathcal{F}_t\right],
$$
where the supremum is taken over the set of adapted processes with
values in $[0,L]$ which satisfy $\int_t^T u(s)ds \leq 1-y$. Here, a
local constraint restricts the rate at which the option can be
exercised to the interval $[0,L]$, while the global constraint
imposes that the maximal volume which can be exercised in the
remaining time from $t$ to $T$ is $1-y$. Then $\bar J(t,y)$ is a
discounted fair price of the swing option contract, if the
expectation is taken with respect to a risk-neutral pricing measure
under which all tradable and storable basic securities in the market
are $\sigma$-martingales.

As the main result of this paper we will show that a `good' version $(J(t,y), \;t\in[0,T],\; y\in(-\infty,1])$ of the adapted random field $(\bar J(t,y), \;t\in[0,T],\; y\in(-\infty,1])$ satisfies the following
first order backward stochastic partial differential equation (BSPDE) in $(t,y)$:
\begin{eqnarray}\label{BSPDE}
 J(t,y)&=&E\left[\left. L \int_t^T (X(s)+D^-_yJ(s,y))_+ds\right|\mathcal{F}_t\right], \nonumber \\
J(T,y)&=&0,\quad J(t,1)=0.
\end{eqnarray}
Here $D^-_yJ$ denotes the left-hand side derivative of $J$ in $y$
and it can be replaced by the right-hand side derivative $D^+_yJ$ in
the above equation. This result will be obtained under the weak
assumptions that $X$ is right-continuous, nonnegative, adapted, and
satisfies some integrability condition. We will also show that
 under these assumptions
$J$ is smooth enough to apply a  variant of a chain rule, which is sufficient to show that a control $u$ is optimal, if and only if
\begin{equation*}
 u(s) \in \left\{\begin{array}{cl} \{0\}, & X(s)+D_y^-J(s,y+\int_t^s  u(r)dr)<0\\ \{L\}, & X(s)+D_y^-J(s,y+\int_t^s u(r)dr)>0 \\
 \ [0,L], &
X(s)+D_y^-J(s,y+\int_t^s  u(r)dr)=0. \end{array}\right.
\end{equation*}
We finally derive a dual minimization problem for $\bar J(t,y)$ in
terms of martingales. This type of dual formulations has its origin
in the pricing problem of American options, see \citet{Ro} and
\citet{HK}, and was later generalized to a pure martingale dual for
multiple exercise options by \citet{Sch} in discrete time and
\citet{Be2} in continuous time. Our dual representation can be seen
as a continuous time version of general dual formulations for
discrete time control problems in \citet{BSS}, \citet{Ro2}, and
\citet{GHW}.

We note that a connection between backward SPDEs and dynamic
programming for a class of non-Markovian control problems was first
studied by \citet{Pe1992}. As in most of the existing  literature
for backward SPDEs he considers parabolic type second order
equations such that the matrix of the higher order coefficients is
positive definite. We also note that some additional conditions on
the coercivity are usually imposed in the literature; see, e.g.,
condition \iindex{(0.4) in \citet{Roz}, Ch. 4}. Without these
conditions, a parabolic type SPDE is regarded as degenerate. For the
degenerate backward SPDEs in the whole space, i.e., without
boundaries, regularity results were obtained in \citet{Roz},
\citet{MaYong}, \citet{Hu}, and more recently by \citet{DTZ} and
\citet{DZ}. The methods developed in these works cannot be applied
in the case of a domain with boundary  because of regularity issues
that prevent using an approximation of the differential operator  by
a non-degenerate one. It turns out that the theory of degenerate
SPDEs in domains is much harder than in the whole space and was, to
the best of our knowledge, not addressed yet in the existing
literature.
 The present paper consider  a problem of
this kind.  We introduce and prove existence for a first order BSPDE
in a domain with boundary. This equation can be interpreted as  a
limit case of a degenerate second order parabolic BSPDE.

The paper is organized as follows: In Section 2 we set the problem and derive some
basic properties of the control problem, including the existence of optimal controls and the construction of the good version $J(t,y)$. In Section 3 we study the marginal values
$-D_y^{\pm} J(t,y)$. It turns out that the left-hand side derivative $D_y^{-}J(t,y)$ in general is a submartingale with rightcontinuous paths, while the right-hand side derivative $D_y^{+}J(t,y)$ may admit discontinuities
from the right. For this reason it is more convenient to work with the left-hand side derivative in most of the proofs. The proof of the main result, namely that $J$ solves the first-order backward
stochastic partial differential equation (\ref{BSPDE}) is given in Section 4. Finally, the characterization of optimal strategies and the dual formulation are presented in Sections 5 and 6.
Uniqueness results for the BSPDE (\ref{BSPDE}) and smoothness of the value process $J(t,y)$ will be discussed in a companion paper, which is in preparation.

\section{Some basic properties of the control problem}

Throughout this paper we assume that $(\Omega, \mathcal{F},
\mathbb{F}, P)$ is a filtered probability space satisfying the usual
conditions and that $(X(t),\; 0\leq t \leq T)$ is a nonnegative,
rightcontinuous, $\mathbb{F}$-adapted stochastic process which
fulfills
\begin{equation}\label{ass:square}
 E[\sup_{0\leq t \leq T} X(t)^p]<\infty
\end{equation}
for some $p>1$.

We consider the following optimization problem: An investor can exercise the cash-flow $X$ continuously, but she is subjected to the constraint
that the rate at which she exercises is bounded by a constant $L>0$, which is fixed from now on. Moreover the maximal total volume of exercise
is bounded by $1$. The investor's aim is to maximize the expected reward, i.e. she wishes to maximize
$$
E\left[\int_0^T u(s)X(s)ds\right]
$$
over all $\mathbb{F}$-adapted processes with values in $[0,L]$ which satisfy $\int_0^T u(s)ds \leq 1$.

In order to study this problem in the general setting introduced
above we consider a dynamic version. For any $[0,T]$-valued stopping
time $\tau$ and $\mathcal{F}_\tau$-measurable $(-\infty,1]$-valued
random variable $Y$ denote by $U(\tau,Y)$ the set of all
$\mathbb{F}$-adapted processes with values in $[0,L]$ such that
$\int_\tau^T u(s) ds \leq 1-Y$. Hence, the investor \iindex{enters}
at time $\tau$ and can spend a total volume of $1-Y$. The
corresponding value of the optimization problem is
$$
\bar J(\tau, Y):=\esssup_{u \in U(\tau,Y)} E\left[\left.\int_\tau^T u(s)X(s)ds \right|\mathcal{F}_\tau\right]
$$
As explained in the introduction, the main result of this paper is that a `good' version $(J(t,y), \;t\in[0,T],\; y\in(-\infty,1])$ of the adapted random field $(\bar J(t,y), \;t\in[0,T],\; y\in(-\infty,1])$ solves
the first order backward stochastic partial differential equation (\ref{BSPDE}).

Before we can prove this result, we need to derive
some basic properties of the corresponding control problem. We first establish existence of optimal controls.

\begin{prop}\label{prop:optimalcontrol}
 For every pair $(\tau,Y)$, where $\tau$ is a stopping time and $Y$ is an $(-\infty,1]$-valued $\mathcal{F}_\tau$-measurable random variable,
there is an optimal control $\bar u \in U(\tau,Y)$, i.e.
$$
\bar J(\tau,Y)=E[\int_\tau^T \bar u(r)X(r)dr|\mathcal{F}_\tau]
$$
\end{prop}
\begin{proof}
 We consider $U(\tau,Y)$ as a subset of $L^q(\mathcal{F}\otimes \mathcal{B}[0,T], P \otimes \lambda_{[0,T]})$ where $\mathcal{B}[0,T]$ and
$\lambda_{[0,T]}$ denote the Borel $\sigma$-field and  the Lebesgue measure on $[0,T]$, and $1/q+1/p=1$ for $p>1$ as in assumption (\ref{ass:square}). Note that $U(\tau,Y)$ is a weakly sequentially compact
subset
of the reflexive Banach space $L^q(\mathcal{F}\otimes \mathcal{B}[0,T], P \otimes \lambda_{[0,T]})$, because $U(\tau,Y)$ is bounded and closed in the strong topology and convex.

We now introduce the set
$$
\mathcal{M}=\left\{ E[\int_\tau^T u(r)X(r)dr|\mathcal{F}_\tau],\; u\in U(\tau,Y) \right\}.
$$
It is straighforward to check that $\mathcal{M}$ is closed under
\iindex{pathwise} maximization, i.e. $M_1, M_2 \in \mathcal{M}$
implies that $M_1 \vee M_2 \in \mathcal{M}$. Hence, by Theorem A.3
in \citet{KS1998}, there is a sequence $(u_n)\subset U(\tau, Y)$
such that
\begin{equation}\label{eq:hilf1}
E[\int_\tau^T  u_n(r)X(r)dr|\mathcal{F}_\tau] \uparrow J(\tau,Y),\quad n \rightarrow \infty.
\end{equation}
As $U(\tau,Y)$ is weakly sequentially compact, we can assume without loss of generality (by passing to a subsequence, if necessary), that $u_n$ converges
weakly in $L^q(\mathcal{F}\otimes \mathcal{B}[0,T], P \otimes \lambda_{[0,T]})$ to some $\bar u \in U(\tau,Y)$. We now show that $\bar u$ is indeed
optimal. Suppose $A\in \mathcal{F}_\tau$.  By weak convergence of $(u_n)$ to $\bar u$ and considering $X  {\bf 1}_{A\times [\tau,T]}$ as an element of
$L^p(\mathcal{F}\otimes \mathcal{B}[0,T], P \otimes \lambda_{[0,T]})$, we get
\begin{equation*}
E[{\bf 1}_A E[\int_\tau^T u_n(r)X(r)dr|\mathcal{F}_\tau]] \uparrow E[{\bf 1}_A E[\int_\tau^T \bar u(r)X(r)dr|\mathcal{F}_\tau]],\quad n\rightarrow \infty,
\end{equation*}
which, combined with (\ref{eq:hilf1}), yields
$$
 E[{\bf 1}_A E[\int_\tau^T \bar u(r)X(r)dr|\mathcal{F}_\tau]]= E[{\bf 1}_A J(\tau,Y)]
$$
As $A\in \mathcal{F}_\tau$ was arbitrary, this immediately gives
$$
J(\tau,Y) = E_\tau[\int_\tau^T  \bar u(r)X(r)dr|\mathcal{F}_\tau].
$$
\end{proof}

At several instances, it will be convenient to switch from the control set $U(\tau,Y)$ to the subset $U'(\tau,Y)$ of controls $u$
which additionally satisfy
\begin{equation}\label{U'}
 u(r)=L \textnormal{ on } \{L(T-r)\leq 1-(Y+\int_\tau^r u(s)ds)\},
\end{equation}
$\lambda_{[0,T]}\otimes P$-almost everywhere.

We collect some facts on the relation between $U(\tau,Y)$ and $U'(\tau,Y)$ in the following proposition.
\begin{prop}\label{prop:U'}
 (i) For any control $u\in U(\tau,Y)$, there is a control $\tilde u\in U'(\tau,Y)$ such that $\tilde u(r)\geq u(r)$ on $[\tau,T]$.
In particular, there exists an optimal strategy $\bar u^{\tau,Y}\in U'(\tau,Y)$ for $J(\tau,Y)$. \\
(ii) $u\in U(\tau,Y)$ belongs to  $ U'(\tau,Y)$, if and only if
$\int_\tau^T u(r)dr=1-Y$ on the set  $\{L(T-\tau)\geq 1-Y\}$ and $u(r)=L$ for $r\in [\tau,T]$ on the set
$\{L(T-\tau)\leq 1-Y\}$.
\end{prop}
\begin{proof}
 For $u\in U(\tau,Y)$ define
$$
\tau_{L,u}=\inf\{r\geq \tau;\; L(T-r) \leq 1- (Y+\int_\tau^r u(s)ds)\}\wedge T.
$$
(i) If $u\in U(\tau,Y)$, then
$$
\tilde u =u {\bf 1}_{[\tau,\tau_{L,u})}+ L {\bf 1}_{[\tau_{L,u},T]} \in U'(\tau,Y).
$$
(ii) Suppose $u\in U'(\tau,Y)$. On the set $\{L(T-\tau)\leq 1-Y\}$ we have $\tau_{L,u}=\tau$. Hence $u(r)=L$ on $[\tau,T]$. On the
set $\{L(T-\tau)\geq 1-Y\}\cap\{\tau_{L,u}<T\}$ we get
\begin{eqnarray}\label{eq:hilf11}
\int_\tau^T u(r)dr &=&\int_\tau^{\tau_{L,u}} u(r)dr + \int_{\tau_{L,u}}^T u(r)dr  \nonumber\\ &=& 1-Y-L(T-\tau_{L,u}) + L( T-\tau_{L,u})=1-Y.
\end{eqnarray}
On the
set $\{L(T-\tau)\geq 1-Y\}\cap\{\tau_{L,u}=T\}$, we obtain $\int_\tau^T u(r)dr \geq 1-Y$ by (\ref{U'}) and the other inequality is trivial.

Now suppose that $u\in U(\tau,Y)$ satisfies the two properties stated in the assertion. If $\tau_{L,u}=\tau$, then $L(T-\tau)\leq 1-Y$
and hence  $u(r)=L$ for $r\in [\tau,T]$. If $\tau<\tau_{L,u}<T$, then $L(T-\tau)> 1-Y$, and hence
$\int_\tau^T u(r)dr=1-Y$. An analogous calculation than in (\ref{eq:hilf11}) shows
$$
\int_{\tau_{L,u}}^T u(r)dr=L(T-\tau_{L,u}),
$$
which implies (\ref{U'}).
\end{proof}

Next, we state the dynamic programming principle for this optimization problem. Its simple proof is omitted.

\begin{prop}\label{prop:dp}
 Suppose $\sigma \leq \tau$ are $[0,T]$-valued stopping times and $Y$ is an $\mathcal{F}_\sigma$-measurable, $(-\infty,1]$-valued random
variable. Then,
$$
\bar J(\sigma,Y)=\esssup_{u\in U(\sigma,Y)} E[\int_\sigma^\tau u(r)X(r)dr + \bar J(\tau,Y+\int_\sigma^\tau u(r)dr)|\mathcal{F}_\sigma]
$$
\end{prop}

The next lemma singles out two properties which are related to Lipschitz continuity and concavity of $J$ in the $y$-variable.
\begin{lem}\label{lem:properties}
 Suppose $\tau$ is a $[0,T]$-valued stopping time and $Y_1,Y_2$ are $\mathcal{F}_\tau$-measurable $(-\infty,1]$-valued random variables.
Then, $P$-almost surely,
\begin{eqnarray}
 |\bar J(\tau,Y_1)-\bar J(\tau,Y_2)| &\leq& E[(\sup_{0\leq t \leq T} X(t))\,|\mathcal{F}_\tau] \,|Y_1-Y_2|. \label{eq:lip} \\
\bar J\left(\tau,\frac{Y_1+Y_2}{2}\right)&\geq& \frac{\bar J(\tau,Y_1)+\bar J(t,Y_2)}{2} \label{eq:concave}
\end{eqnarray}
\end{lem}
\begin{proof}
 We first show (\ref{eq:lip}). Choose an optimal strategy $\bar u^{\tau,Y_1}\in U(\tau,Y_1)$ and define
$\sigma=\inf\{t\geq \tau;\; \int_\tau^t \bar u^{\tau,Y_1}(s)ds\geq 1-Y_2\}\wedge T$. Then,
$u(t)=\bar  u^{\tau,Y_1} {\bf 1}_{[\tau,\sigma]}\in U(\tau,Y_2)$. Consequently, on the set $\{Y_1\leq Y_2\}$,
we get
\begin{eqnarray*}
0&\leq& \bar J(\tau,Y_1)-\bar J(\tau,Y_2)\leq E[\int_\tau^T (\bar u^{\tau,Y_1}(s)-u(s))X(s)ds|\mathcal{F}_\tau] \\&=& E[\int_\sigma^T
\bar u^{\tau,Y_1}(s)X(s)ds|\mathcal{F}_\tau] \leq E[(\sup_{0\leq r \leq T} X(r)) \int_\sigma^T
\bar u^{\tau,Y_1}(s)ds|\mathcal{F}_\tau] \\ &\leq&  E[(\sup_{0\leq r \leq T} X(r)) (Y_2-Y_1)|\mathcal{F}_\tau]
\end{eqnarray*}
Changing the roles of $Y_1$ and $Y_2$, \iindex{we obtain that} this
inequality also holds on $\{Y_1> Y_2\}$, which proves
(\ref{eq:lip}).

For (\ref{eq:concave}) one merely needs to note that for $u_1\in U(\tau, Y_1)$ and $u_2 \in U(\tau, Y_2)$, the control
$(u_1+u_2)/2$ belongs to $U(\tau, (Y_1+Y_2)/2)$.
\end{proof}

We next construct a `good' version of $\bar J(t,y)$ as stated in the following proposition.

\begin{prop}\label{prop:goodversion}
 There is an adapted random field $(J(t,y),\; t\in[0,T],\, y\in(-\infty,1])$ with the following properties:
\\ a) For every pair $(\tau,Y)$
$$
J(\tau,Y)=\bar J(\tau,Y)\quad P-a.s.
$$
In particular, for every $y\in (-\infty ,1]$, $J(t,y)$ is an adapted modification of $\bar J(t,y)$.
\\ b) There is a set $\bar \Omega \in \mathcal{F}$ with $P(\bar \Omega)=1$ such that the following properties hold on $\bar \Omega$:
\begin{enumerate}
 \item For every $y\in(-\infty,1]$, the mapping $t\mapsto  J(t,y)$ is RCLL.
  \item For every $t\in[0,T]$ and $y_1,y_2\in (-\infty,1]$
 $$
|J(t,y_1)-J(t,y_2)|\leq \left(\sup_{r\in [0,T]} Z(r) \right) |y_1-y_2|
$$
where $Z(t)$ is a RCLL modification of $E[\sup_{r\in [0,T]} X(r)|\mathcal{F}_t]$ which satisfies $\sup_{r\in [0,T]} Z(r) < \infty$
on $\bar \Omega$.
\item For every $t\in[0,T]$, the mapping $y\mapsto J(t,y)$ is concave.
\end{enumerate}

\end{prop}

As a preparation we need the following lemma.
\begin{lem}\label{lem:rc1}
(a) For every $y\in(-\infty,1]$, the mapping $t\mapsto E[\bar J(t,y)]$ is rightcontinuous. \\ (b) For every $y\in(-\infty,1]$,
the process $\bar J(t,y)$ has a modification $\hat J(t,y)$, which is a supermartingale whose paths are RCLL with probability one.
\end{lem}
\begin{proof}
We fix some $y\in (-\infty,1]$. Notice first that $\bar J(t,y)$ is a supermartingale on $[0,T]$, because $U(t_2,y)\subset U(t_1,y)$ for $0\leq t_1\leq t_2\leq T$.
Hence, by Theorem 1.3.13 in \cite{KS1991}, (a) implies (b). For (a) we fix some $t\in[0,T)$ and choose a sequence $(t_n)\subset [0,T]$ such that
$t_n \downarrow t$. By the supermartingale property we have $E[\bar J(t,y)]\geq E[\bar J(t_n,y)]$. So it is sufficient to show that
$\liminf_{n\rightarrow \infty} E[\bar J(t_n,y)] \geq E[\bar J(t,y)]$. To this end we choose an optimal strategy $\bar u^{t,y}\in U(t,y)$ for $\bar J(t,y)$. Then,
$u_n=\bar u^{t,y} {\bf 1}_{[t_n,T]}\in U(t_n,y)$. Therefore, by dominated convergence,
$$
\liminf_{n\rightarrow \infty} E[\bar J(t_n,y)]\geq \liminf_{n\rightarrow \infty} E[\int_{t_n}^T \bar u^{t,y}(s)ds]=E[\int_{t}^T \bar u^{t,y}(s)ds]=\bar J(t,y).
$$
\end{proof}

\begin{proof}[Proof of Proposition \ref{prop:goodversion}.]
Let
$$
Q_1:=([0,T]\cap \mathbb{Q})\cup \{T\},\quad Q_2:=(-\infty,1] \cap \mathbb{Q}
$$
We choose a set $\bar \Omega$ with $P(\bar \Omega)=1$ such that the following properties hold on $\bar \Omega$:
\begin{itemize}
 \item[(i)] $Z^*:=\sup_{r\in [0,T]} Z(r)<\infty$.
  \item[(ii)] $\hat J(t,y)=\bar J(t,y)$ for every $(t,y)\in Q_1\times Q_2$ (where $\hat J$ was constructed in the previous lemma).
\item[(iii)] The mapping $t\mapsto \hat J(t,y)$ is RCLL for every $y\in Q_2$.
 \item[(iv)] For every $(t,y_1,y_2)\in Q_1\times Q_2^2$ it holds that
$$
|\bar J(t,y_1)-\bar J(t, y_2)|\leq Z^* |y_1-y_2|.
$$
\item[(v)] For every $(t,y_1,y_2)\in Q_1\times Q_2^2$ it holds that
$$
\bar J(t,\frac{y_1+y_2}{2})\geq \frac{\bar J(t,y_1)+\bar J(t,y_2)}{2}.
$$
\end{itemize}
We briefly check that such a set $\bar \Omega$ exists. The martingale $E[\sup_{0\leq r \leq T} X(r)|\mathcal{F}_t]$ has an RCLL
modification which we denote $Z(t)$. By Doob's inequality it satisfies
$$
E[\sup_{0\leq t\leq T} Z(t)^p]\leq \left(\frac{p}{p-1}\right)^p E[Z(T)^p] =\left(\frac{p}{p-1}\right)^p E[\sup_{0\leq t \leq T} X(t)^p]<\infty.
$$
Hence, the random variable $Z^*$ is almost surely finite. Moreover, (ii) and (iii) can be realized by the previous lemma, because
$Q_1$ and $Q_2$ are countable. The same applies to (iv) and (v) in view of Lemma \ref{lem:properties}.

On $\bar \Omega$ we wish to define $J(t,y)$ in the following way:
In a first step we define $J(t,y)=\hat J(t,y)$ for $(t,y)\in Q_1\times Q_2$. In a second step we let
$$
J(t,y)= \lim_{Q_2\ni\tilde y \rightarrow y} \hat J( t, \tilde y)
$$
for $t\in Q_1, y \in (-\infty,1] \setminus Q_2$. Then, $J(t,y)$ is defined on $Q_1\times (-\infty,1]$. In the final step we set
$$
J(t,y)= \lim_{Q_1\ni\tilde t \downarrow t}  J(\tilde t,y)
$$
for $t\in [0,T]\setminus Q_1$ and $y\in [0,1]$.

So we first have to show that the limits in the above construction exist on $\bar \Omega$. Fix $t \in Q_1$ and $y \in (-\infty,1]\setminus Q_2$. We choose a sequence
$(\tilde y_n)\subset Q_2$ such that $\tilde y_n\rightarrow y$. Then, by (ii) and (iv),
$$
|\hat J(t, \tilde y_n)-\hat J( t, \tilde y_m)|\leq Z^* |\tilde y_n -\tilde y_m|
$$
In view of (i), $(\hat J( t, \tilde y_n))$ is a Cauchy sequence and, as its limit does certainly not depend on the choice of the sequence,
we see that $\lim_{Q_2\ni\tilde y \rightarrow y} \hat J( t, \tilde y)$ exists. Hence $J(t,y)$ is well-defined on $Q_1\times (-\infty,1]$. Moreover, it is straightforward to check that
for $t\in Q_1$ and $(y_1,y_2)\in (-\infty,1]^2$
$$
|J(t,y_1)-J(t,y_2)| \leq Z^* |y_1-y_2|
$$
holds true.

Now we fix some $t\in [0,T]\setminus Q_1$ and some $y\in (-\infty,1]$. We choose sequences $(t_n)\subset Q_1$ and $(y_k)\subset Q_2$ such
that $t_n \downarrow t$ and $y_k\rightarrow y$. Then,
\begin{eqnarray}\label{eq:hilf3}
 && |J(t_n,y)-J(t_m,y)| \nonumber\\
&\leq & |\hat J(t_n,y_k)-\hat J(t_m,y_k)| + |J(t_n,y)-J(t_n,y_k)| + |J(t_m,y)-J(t_m,y_k)| \nonumber \\
&\leq & |\hat J(t_n,y_k)-\hat J(t_m,y_k)| + 2 Z^* |y-y_k|
\end{eqnarray}
 By (iii) we can conclude that the sequence $(J(t_n,y))$ is Cauchy, and, hence, $\lim_{Q_1\ni\tilde t \downarrow t} \bar J(\tilde t,y)$
exists, because the limit does not depend on the approximating sequence. So $J$ is well-defined.

We now prove that $J$ satisfies the properties stated in b) on $\bar
\Omega$. The Lipschitz property b2) can be immediately transferred
from $t\in Q_1$ (for which it was shown above)  to general $t$ by
the construction of $J$. Property b1), which states that $J$ has
RCLL paths in $t$, can be shown by a similar argument as in
(\ref{eq:hilf3}). It remains to show concavity in $y$. As $J(t,y)$
is continuous in $y$ for fixed $t$, it is sufficient to show that
$$
J(t,\frac{y_1+y_2}{2})\geq \frac{J(t,y_1)+J(t,y_2)}{2}
$$
holds for every $(t,y)\in [0,T]\times (-\infty,1]$. By (v) and the construction of $J$, it is valid for $(t,y)\in Q_1\times Q_2$. By the continuity
properties of $\bar J$ this immediately extends to general $(t,y)$.

It remains to prove a). Suppose $\tau$ is a $[0,T]$-valued stopping time and $Y$ is a $\mathcal{F}_\tau$ measurable, $(-\infty,1]$-valued random variable. We
can approximate $Y$ by a nonincreasing sequence $(Y_n)$ of $Q_2$-valued, $\mathcal{F}_\tau$ measurable random variables. Moreover, we can choose a sequence
$(\tau_n)$ of $Q_1$-valued stopping times such that $\tau_n$ converges nonincreasingly to $\tau$. By (ii) and the continuity properties b1) and b2)
of $J$, we get on $\bar \Omega$
\begin{displaymath}
  J(\tau,Y)=\lim_{n\rightarrow \infty} J(\tau_n, Y_n)= \lim_{n \rightarrow \infty} \bar J(\tau_n, Y_n).
\end{displaymath}
So it remains to show that
\begin{equation}\label{eq:hilf4}
\lim_{n \rightarrow \infty} \bar J(\tau_n, Y_n)=\bar J(\tau,Y),\quad P-a.s.
\end{equation}
As $U(\tau_n,Y_n)\subset U(\tau,Y)$ we observe that
$$
E[\bar J(\tau_n,Y_n)|\mathcal{F}_{\tau}] \leq \bar J(\tau,Y)
$$
Hence, we have that for every $A\in \mathcal{F}_\tau$
\begin{equation*}
 E[{\bf 1}_A \bar J(\tau_n,Y_n)] \leq E[{\bf 1}_A \bar J(\tau,Y)],
\end{equation*}
which implies
\begin{equation}\label{eq:hilf5}
 \limsup_{n\rightarrow \infty}  \bar J(\tau_n,Y_n) \leq \bar J(\tau,Y)
\end{equation}
Now choose some optimal strategy $\bar u^{\tau,Y} \in U(\tau,Y)$ for $\bar J(\tau,Y)$ and define
$$
u_n(t)={\bf 1}_{[\tau_n,\sigma_n]}(t) \bar u^{\tau,Y}(t)
$$
where
$$
\sigma_n=\inf\{t\geq \tau_n;\; \int_{\tau}^t \bar u^{\tau,Y}(s) ds \geq 1-Y_n\}\wedge T.
$$
Then $u_n \in U(\tau_n,Y_n)$. As $Y_n$ is nonincreasing, the sequence of stopping times $\sigma_n$ is nondecreasing. Denoting its limit by
$\sigma$ we obtain that
$$
\int_{\tau}^\sigma \bar u^{\tau,Y}(s) ds =1-Y
$$
or $\sigma=T$. As  $\bar u^{\tau,Y} \in U(\tau,Y)$, we have in any case that
$$
\int_\sigma^T \bar u^{\tau,Y}(s)ds=0.
$$
Thus,
\begin{eqnarray*}
 && E[\bar J(\tau_n,Y_n)|\mathcal{F}_{\tau}]\geq  E[\int_{\tau_n}^T u_n(s)X(s)ds|\mathcal{F}_{\tau}] \\
&=& \bar J(\tau,Y)-E[\int_\tau^{\tau_n} \bar u^{\tau,Y}(s)X(s)ds|\mathcal{F}_{\tau}] -E[\int_{\sigma_n}^\sigma \bar u^{\tau,Y}(s)X(s)ds|\mathcal{F}_{\tau}].
\end{eqnarray*}
As $\tau_n\downarrow \tau$ and $\sigma_n\uparrow \sigma$, the dominated convergence theorem yields for every $A\in\mathcal{F}_\tau$
$$
\liminf_{n\rightarrow \infty} E[{\bf 1}_A \bar J(\tau_n,Y_n)] \geq E[{\bf 1}_A \bar J(\tau,Y)],
$$
which in turn implies
$$
\liminf_{n\rightarrow \infty}  \bar J(\tau_n,Y_n) \geq \bar J(\tau,Y)
$$
and in view of (\ref{eq:hilf5}) finishes the proof of (\ref{eq:hilf4}).
\end{proof}

\section{Some properties of the marginal value}

By Proposition \ref{prop:goodversion}, there is a set $\Bar \Omega$ of full measure such that for all $(t,y)$ the left-hand side
derivative $D^-_y J(t,y)$ and the right-hand side derivative $D^+_y J(t,y)$ in $y$-direction exist on $\bar \Omega$ due to concavity.
In order to study the  marginal values  $-D^-_y J(t,y)$ and $-D^+_y J(t,y)$, we first derive some properties
related to the difference process $J(t,y+h)-J(t,y)$.

\begin{prop}\label{prop:marginal}
 Suppose $\tau$ is a $[0,T]$-valued stopping time and $Y_2\geq Y_1$ are $\mathcal{F}_\tau$-measurable, $(-\infty,1]$-valued random variables. Denote
by $\bar u^{\tau,Y_1}\in U'(\tau, Y_1)$, $\bar u^{\tau,Y_2}\in U'(\tau,Y_2)$ optimal controls for $\bar J(\tau,Y_1)$ and  $\bar J(\tau,Y_1)$, respectively.
Then,
\\ (i) It holds that
$$
\bar J(\tau,Y_1)-\bar J(\tau,Y_2)=\esssup_{u \in \tilde U(\bar u^{\tau,Y_2},Y_2-Y_1)} E[\int_t^T u(r)X(r) dr|\mathcal{F}_\tau],
$$
where $\tilde U(\bar u^{\tau,Y_2},Y_2-Y_1)$ denotes the set of adapted processes $u$ such that $\int_\tau^T u(r)dr \leq Y_2-Y_1$ and
$0\leq u(r)\leq L-\bar u^{\tau,Y_2}(r)$ for $r\in [\tau,T]$.
\\ (ii) Define $\bar u$ on $[\tau,T]$ by
$$
\bar u(r)= (\bar u^{\tau,Y_1}(r)-\bar u^{\tau,Y_2}(r))_+{\bf 1}_{\{r;\; \int_\tau^r
(\bar u^{\tau,Y_1}(s)-\bar u^{\tau,Y_2}(s))_+ ds\leq Y_2-Y_1\}}
$$
Then, $\bar u\in \tilde U(\bar u^{\tau,Y_2},Y_2-Y_1)$. Moreover, $\bar u+  \bar u^{\tau,Y_2}\in U'(\tau,Y_1)$ and is an optimal control for $\bar J(\tau,Y_1)$.
\end{prop}
\begin{proof}
 We prove both items at the same time. Let $u\in \tilde U(\bar u^{\tau,Y_2},Y_2-Y_1)$. Then, it is straightforward to check
that $ \bar u^{\tau,Y_2}+u\in U(\tau,Y_1)$. Hence,
$$
\bar J(\tau,Y_1)\geq E[\int_\tau^T (\bar u^{\tau,Y_2}(s)+u(s))X(s)ds|\mathcal{F}_\tau]=\bar J(\tau, Y_2)+E[\int_\tau^T u(s)X(s)ds|\mathcal{F}_\tau].
$$
This implies
\begin{equation}\label{eq:hilf6}
\bar J(\tau,Y_1)-\bar J(\tau,Y_2)\geq\esssup_{u \in \tilde U(\bar u^{\tau,Y_2},Y_2-Y_1)} E[\int_t^T u(r)X(r) dr|\mathcal{F}_\tau].
\end{equation}
We will show that $\bar u$, defined in (ii), satisfies
\begin{eqnarray}
 \bar u &\in& \tilde U(\bar u^{\tau,Y_2},Y_2-Y_1) \label{eq:hilf7}\\
 \bar u^{\tau,Y_1}- \bar u&\in& U(\tau, Y_2) \label{eq:hilf8}
\end{eqnarray}
This proves (i), because
\begin{eqnarray*}
 \bar J(\tau,Y_1)&=&E[\int_\tau^T \bar u^{\tau,Y_1}(s)X(s)ds|\mathcal{F}_\tau] \\
&=& E[\int_\tau^T (\bar u^{\tau,Y_1}(s)-\bar u(s))X(s)ds|\mathcal{F}_\tau] +E[\int_t^T \bar u(s)X(s) ds|\mathcal{F}_\tau]
\\ &\leq & \bar J(\tau,Y_2)+ \esssup_{u \in \tilde U(\bar u^{\tau,Y_2},Y_2-Y_1)} E[\int_t^T u(r)X(r) dr|\mathcal{F}_\tau]
\end{eqnarray*}
In view of (\ref{eq:hilf6}), the inequality turns into an identity. Hence we obtain (i) and the optimality of $\bar u$ for the problem $\bar J(\tau,Y_1)-\bar J(\tau,Y_2)$.  This implies
optimality of $\bar u+  \bar u^{\tau,Y_2}$ for $\bar J(\tau,Y_1)$, because
$$
\bar J(\tau,Y_1)=\bar J(\tau,Y_2)+E[\int_t^T \bar u(s)X(s) ds|\mathcal{F}_\tau]=E[\int_t^T (\bar u^{\tau,Y_2}(s)+\bar u(s))X(s) ds|\mathcal{F}_\tau].
$$
We will now  verify (\ref{eq:hilf7}) and (\ref{eq:hilf8}). Notice first that (\ref{eq:hilf7}) is rather obvious, because
$$
\int_\tau^T \bar u(r)dr \leq Y_2-Y_1
$$
by construction, and
$$
0\leq \bar u(r) \leq \bar u^{\tau,Y_1}(r)-\bar u^{\tau,Y_2}(r) \leq L-\bar u^{\tau,Y_2}(r)
$$
for $r\in [\tau,T]$.

We prove (\ref{eq:hilf8}) on the sets  $\{L(T-\tau)\leq 1-Y_2\}$,  $\{L(T-\tau)>1-Y_2\}\cap \{L(T-\tau)\leq 1-Y_1\}$ and $\{L(T-\tau)>1-Y_1\}$ separately.
On the set $\{L(T-\tau)\leq 1-Y_2\}$, we get $\bar u^{\tau,Y_1}(r)=\bar u^{\tau,Y_2}(r)=L$  for $r\in [\tau,T]$ by Proposition \ref{prop:U'}. Hence, $\bar u(r)=0$ and $\bar u^{\tau,Y_1}-\bar u= \bar u^{\tau,Y_2}\in U(\tau, Y_2)$.
On the set $\{L(T-\tau)>1-Y_2\}\cap \{L(T-\tau)\leq 1-Y_1\}$, we get $\bar u^{\tau,Y_1}(r)=L$  for $r\in [\tau,T]$ and $\int_\tau^T \bar u^{\tau,Y_2}(r)dr=1-Y_2$ by Proposition \ref{prop:U'}.
Hence, we obtain on this set, for every $r\in [\tau,T]$,
$$
\int_\tau^r (\bar u^{\tau,Y_1}(s)-\bar u^{\tau,Y_2}(s))_+ ds \leq L(T-\tau)-\int_\tau^T \bar u^{\tau,Y_2}(s) ds \leq 1-Y_1-(1-Y_2)=Y_2-Y_1.
$$
This again implies $\bar u^{\tau,Y_1}-\bar u= \bar u^{\tau,Y_2}\in U(\tau, Y_2)$. On the set
$\{L(T-\tau)>1-Y_1\}$, we already know, by Proposition \ref{prop:U'}, that
$$
\int_\tau^T \bar u^{\tau,Y_1}(s)ds=1-Y_1,\quad  \int_\tau^T \bar u^{\tau,Y_2}(s)ds=1-Y_2.
$$
Hence,
$$
\int_\tau^T (\bar u^{\tau,Y_1}(s)-\bar u^{\tau,Y_2}(s)) ds = 1-Y_1-(1-Y_2)=Y_2-Y_1,
$$
which yields
$$
\int_\tau^T \bar u(s)ds= Y_2-Y_1.
$$
Consequently,
$$
\int_\tau^T (\bar u^{\tau,Y_1}(s)-\bar u(s))ds =1-Y_1-(Y_2-Y_1)=1-Y_2
$$
Moreover,
$$
0\leq \min\{\bar u^{\tau,Y_2}(r), \bar u^{\tau,Y_1}(r)\}  \leq \bar u^{\tau,Y_1}(r)-\bar u(r) \leq L
$$
for $r\in [\tau,T]$. So, $u^{\tau,Y_1}-\bar u \in U(\tau,Y_2)$ also holds on $\{L(T-\tau)>1-Y_1\}$.

By the arguments in the proof of (\ref{eq:hilf8}) it is easy to see that $\bar u+\bar u^{\tau,Y_2}\in U'(\tau,Y_1)$ thanks to by Proposition \ref{prop:U'}.
\end{proof}

\begin{cor}
 Suppose $\sigma\leq \tau$  are $[0,T]$-valued stopping times and $Y$ is an $\mathcal{F}_\sigma$-measurable random variable with values in $(-\infty,1]$. Then there
are  optimal controls $\bar u^{\tau,Y}$ for $\bar J(\tau,Y)$ and $\bar u^{\sigma,Y}$ for $\bar J(\sigma,Y)$
such that $\bar u^{\tau,Y}(r) \geq \bar u^{\sigma,Y}(r)$ for $r\in [\tau,T]$. Moreover, $\bar u^{\tau,Y}$ can be chosen from the set $U'(\tau,Y)$.
\end{cor}
\begin{proof}
 Choose optimal controls $\bar u^{\sigma,Y}\in U'(\sigma,Y)$ for $\bar J(\sigma,Y)$ and $\bar u^{\tau,\tilde Y}\in U'(\tau,\tilde Y)$ for $\bar J(\tau,\tilde Y)$,
where $\tilde Y=Y+\int_\sigma^\tau u^{\sigma,Y}(r)dr$. By the
dynamic programming principle in Proposition \ref{prop:dp}, we
observe that
$$
\bar u^{\sigma,Y}=u^{\sigma,Y} {\bf 1}_{[\sigma,\tau)}+u^{\tau,\tilde Y}{\bf 1}_{[\tau,T)}
$$
is also optimal for $\bar J(\sigma,Y)$. As $\tilde Y\geq Y$, part (ii) of the  previous proposition implies that there is an optimal control
$\bar u^{\tau,Y}\in U'(\tau,Y)$ for $J(\tau,Y)$ such that $\bar u^{\tau,Y}(r) \geq  u^{\tau,\tilde Y}(r)$ for $r\in [\tau,T]$.
\end{proof}

The following proposition includes as a special case the statement that the difference process $\bar J(t,y+h)-\bar J(t,y)$ is
a submartingale for every $y\in (-\infty,1]$ and $h\in [0,1-y]$.

\begin{prop}\label{prop:submartingale}
 Suppose $\sigma\leq \tau$  are $[0,T]$-valued stopping times and $Y_1\leq Y_2$ are $\mathcal{F}_\sigma$-measurable, $(-\infty,1]$-valued random variables.
Then,
$$
E[\bar J(\tau,Y_2)-\bar J(\tau,Y_1)|\mathcal{F}_\sigma]\geq \bar J(\sigma,Y_2)-\bar J(\sigma,Y_1).
$$
\end{prop}
\begin{proof}
 By the previous corollary, we can choose optimal controls $\bar u^{\tau,Y_2}$ for $\bar J(\tau,Y_2)$ and $\bar u^{\sigma,Y_2}$ for $\bar J(\sigma,Y_2)$
such that $\bar u^{\tau,Y_2}(r) \geq \bar u^{\sigma,Y_2}(r)$ for $r\in [\tau,T]$ and $\bar u^{\tau,Y_2}\in U'(\tau,Y_2)$. Moreover, by Proposition \ref{prop:marginal}
we can choose $\bar u^{\tau,Y_1}$ optimal for $J(\tau,Y_1)$ such that $\bar u^{\tau,Y_1}- \bar u^{\tau,Y_2}\in \tilde U(\bar u^{\tau,Y_2},Y_2-Y_1)$.
Consequently,
$$
u(r):=\bar u^{\sigma,Y_2}(r)+ {\bf 1}_{[\tau,T]}(r) (\bar u^{\tau,Y_1}(r)- \bar u^{\tau,Y_2}(r))
$$
belongs to $U(\sigma,Y_1)$. This yields
\begin{eqnarray*}
\bar J(\sigma,Y_1)&\geq& E[\int_\sigma^T u(s)X(s)ds|\mathcal{F}_\sigma] \\ &=& E[\int_\sigma^T \bar u^{\sigma,Y_2}(s)X(s)ds|\mathcal{F}_\sigma]
+  E[\int_\tau^T \bar u^{\tau,Y_1}(s)X(s)ds|\mathcal{F}_\sigma] \\ && - E[\int_\tau^T \bar u^{\tau,Y_2}(s)X(s)ds|\mathcal{F}_\sigma] \\
&=& \bar J(\sigma, Y_2) + E[\bar J(\tau, Y_1)|\mathcal{F}_\sigma] -E[\bar J(\tau, Y_2)|\mathcal{F}_\sigma]
\end{eqnarray*}

\end{proof}

In view of Proposition \ref{prop:goodversion} and \ref{prop:submartingale} we immediately obtain the following result. It states that the marginal values
$-D^{\pm}_y J(t,y)$ are supermartingales, analogously to the situation for discrete time multiple stopping problems in \citet{MH} and \citet{Be1}.

\begin{cor}\label{cor:submartingale}
 (i) For every $y\in (-\infty,1]$, the left-hand side derivative $D^-_yJ(t,y)$ is a submartingale. \\
(ii) For every $y\in (-\infty,1)$, the right-hand side derivative $D^+_yJ(t,y)$ is a submartingale.
\end{cor}

We will now study the regularity of the one-sided derivatives $D^-_y  J(t,y)$ and $D^+_y   J(t,y)$ as processes in time.
 The following example is instructive to see what kind of results we can expect.
\begin{exmp}\label{exmp:nondiff}
 Suppose $\rho$ is a stopping time of the filtration $\mathbb{F}$ with values in $[0,T]$
and consider the RCLL process
$$
X(t)={\bf 1}_{[0,\rho)}(t)
$$
Then, certainly it is optimal to exercise as soon as possible, i.e. $\bar u^{t,y}=L{\bf 1}_{[t,t+(1-y)/L]}$ is an optimal control for
$J(t,y)$. Therefore,
$$
J(t,y)=E[\min(1-y, L(\rho-t))|\mathcal{F}_t]{\bf 1}_{\{\rho\geq t\}}.
$$
Thus, the one sided derivatives of $J$ are
\begin{eqnarray*}
 D^-_y J(t,y)&=& -E[{\bf 1}_{\{y> 1-L(\rho-t)\}}|\mathcal{F}_t]\\
D^+_y J(t,y)&=& -E[{\bf 1}_{\{y\geq 1-L(\rho-t)\}}|\mathcal{F}_t].
\end{eqnarray*}
It follows that
$$
E[D^+_y J(t,y)]=P(\{\rho<(1-y)/L+t\}-1.
$$
If the distribution function of $\rho$ has a jump at $(1-y_0)/L+t_0$, then the mapping $t\mapsto E[D^+_y J(t,y_0)]$ is not rightcontinuous
at $t_0$. This implies that $D^+_y J(t,y)$ does not admit a rightcontinuous version in $t$, if the distribution function of $\rho$ is
discontinuous. Contrarily
$$
E[D^-_y J(t,y)]=P(\{\rho\leq(1-y)/L+t\}-1
$$
is rightcontinuous in $t$ for every $y$. As $D^-_y J(t,y)$ is a submartingale for fixed $y$, we conclude, that, for every $y$,
$D^-_y J(t,y)$ has an RCLL modification.
\end{exmp}

\begin{prop}\label{prop:derivative}
 (i) For every $y\in (-\infty,1]$, the submartingale $D^-_y J(t,y)$ has an RCLL modification. \\
(ii)  $\lambda_{[0,T]}\otimes P(\{D^-_y J(\cdot,y)\neq D^+_y J(\cdot,y)\})=0$  for $\lambda_{(-\infty,1)}$-a.e. $y$.
\end{prop}
\begin{proof}
Notice first, that $\bar J(t,y)=E[L\int_t^T X(s)ds|\mathcal{F}_t]$ for $y\leq 1-LT$. Hence, $D^\pm_y J(t,y)=0$ for $y<1-LT$. We can hence restrict ourselves to $y\in [1-LT,1]$ for the rest of the proof.
 \\
(i) $D^-_y J(t,y_0)$ is a submartingale by Corollary
\ref{cor:submartingale} for every $y_0$. Hence, it is sufficient to
prove that for every $y_0$, the mapping $t\rightarrow E[D^-_y
J(t,y_0)]$ is rightcontinuous. Fix  $t\in[0,T)$ and a sequence
$\Delta_n\downarrow 0$. By the submartingale property, $E[D^-_y
J(t+\Delta_n,y)]\geq E[D^-_y J(t,y)]$ is nonincreasing. By the
concavity of $J(t,y)$ in $y$, we hence obtain
\begin{eqnarray*}
 && \int_{1-LT}^1 |E[D^-_y J(t+\Delta_n,y)]-  E[D^-_y J(t,y)]|dy \\ &=& E\left[\int_{1-LT}^1 (D^-_y J(t+\Delta_n,y)-  D^-_y J(t,y))dy \right] \\ &=&E[J(t,1-LT)-J(t+\Delta_n,1-LT)] \\&\rightarrow& 0
\end{eqnarray*}
for $n \rightarrow \infty$ by the rightcontinuity of $J(t,1-LT)$ in $t$. Thus, for almost every $y$,
\begin{equation}\label{eq:RC}
|E[D^-_y J(t+\Delta_n,y)]-  E[D^-_y J(t,y)]|\rightarrow 0,\quad n\rightarrow \infty.
\end{equation}
Now fix some arbitrary $y_0$ and choose a sequence $y_k\uparrow y_0$ such that (\ref{eq:RC}) holds for every $y_k$. Note that by concavity,
$$
E[D^-_y J(t+\Delta_n,y_0)]\leq E[D^-_y J(t+\Delta_n,y_k)].
$$
Consequently,
\begin{eqnarray*}
 0 &\leq&  E[D^-_y J(t+\Delta_n,y_0)]-  E[D^-_y J(t,y_0)]\\&\leq& E[D^-_y J(t+\Delta_n,y_k)]-  E[D^-_y J(t,y_k)] +E[D^-_y J(t,y_k)]-  E[D^-_y J(t,y_0)].
\end{eqnarray*}
 By (\ref{eq:RC}) we thus obtain
$$
\limsup_{n\rightarrow \infty} E[D^-_y J(t+\Delta_n,y_0)]-E[D^-_y J(t,y_0)] \leq E[D^-_y J(t,y_k)]-  E[D^-_y J(t,y_0)].
$$
Letting $k$ tend to infinity we observe that
$$
\lim_{n\rightarrow \infty} E[D^-_y J(t+\Delta_n,y_0)]=E[D^-_y J(t,y_0)].
$$
\\
(ii) We define the measurable set
$$
\mathcal{N}:=\{((t,\omega,y)\in[0,T]\times\Omega\times[1-LT,1];\; D^-_yJ(t,y,\omega)\neq D^-_yJ(t,y,\omega)\}
$$
and consider the sections
\begin{eqnarray*}
 \mathcal{N}_y&=&\{(t,\omega)\in[0,T]\times\Omega;\; D^-_yJ(t,y,\omega)\neq D^-_yJ(t,y,\omega)\},\quad y\in[1-LT,1] \\
 \mathcal{N}_{(t,\omega)}&=&\{y\in [1-LT,1];\; D^-_yJ(t,y,\omega)\neq D^-_yJ(t,y,\omega)\},\quad (t,\omega) \in [0,T]\times \Omega.
\end{eqnarray*}
It is sufficient to show that
$$
\int_{1-LT}^1  (\lambda_{[0,T]} \otimes P)(\mathcal{N}_y) dy =0.
$$
By Fubini's theorem
\begin{eqnarray*}
 \int_{1-LT}^1  (\lambda_{[0,T]} \otimes P)(\mathcal{N}_y) dy= \int_{[0,T]\times \Omega} \lambda_{[1-LT,1]} (\mathcal{N}_{(t,\omega)}) d(\lambda_{[0,T]}\otimes P).
\end{eqnarray*}
However, $ \lambda_{[1-LT,1]} (\mathcal{N}_{(t,\omega)})=0$ for every $(t,\omega)\in [0,T]\times \bar\Omega$, (where $\bar \Omega$ is the set of full measure constructed in Proposition \ref{prop:goodversion}),
by concavity of the function $y\mapsto J(t,\omega,y)$.
\end{proof}

\section{Existence for the BSPDE}

In this section we prove that the good version of the value process $J(t,y)$ indeed solves the  BSPDE (\ref{BSPDE}).

\begin{thm}
 For every $y\in (-\infty,1)$ and $t\in [0,T]$
\begin{eqnarray*}
 J(t,y)&=&E\left[\left. L \int_t^T (X(s)+D^-_yJ(s,y))_+ds\right|\mathcal{F}_t\right],  \\
J(t,1)&=&0
\end{eqnarray*}
holds $P$-almost surely. Moreover, the left-hand side derivative $D^-_y$ can be replaced by the right-hand side derivative $D^+_y$.
\end{thm}
\begin{proof}
 The boundary condition $J(t,1)=0$ is obviously satisfied. \\
{\it Step 1:} We show for every $y\in(-\infty,1)$
$$
J(t,y)\leq E\left[\left. L \int_t^T (X(s)+D^-_yJ(s,y))_+ds\right|\mathcal{F}_t\right].
$$
To this end we first fix some $(t,y)\in[0,T]\times (-\infty,1)$ and
choose a sequence of partitions $\pi_n=\{t^n_0,t^n_1,\ldots,t^n_n\}$
of $[t,T]$ such that the mesh size $|\pi_n|=\max_{i=1,\ldots, n}
|t^n_i-t^n_{i-1}|$ tends to zero and with $t^n_0=t$ and $t_n^n=T$.
We denote by $\bar u^{t_i^n,y}$ an optimal control for $\bar
J(t_i^n,y)$ and define $\bar Y_i^{n}:=\int_{t_i^n}^{t_{i+1}^n} \bar
u^{t_i^n,y}(r)dr$. Applying the dynamic programming principle
(Proposition \ref{prop:dp}) repeatedly, we obtain
\begin{eqnarray*}
 && J(t,y)\\&=&E[\int_{t_0^n}^{t_1^n} \bar u^{t_0^n,y}(r)X(r)dr + J(t_1,y) + (J(t_1,y+\bar Y_0^n)-J(t_1,y))|\mathcal{F}_t] \\
&=& \sum_{i=0}^{n-1} E[\int_{t_i^n}^{t_{i+1}^n} \bar u^{t_i^n,y}(r)X(r)dr |\mathcal{F}_t]  \\ && +  \sum_{i=0}^{n-1} E[ J(t_{i+1},y+\bar Y_i^n)-J(t_{i+1},y)|\mathcal{F}_t] \\
&=& \sum_{i=0}^{n-1} E[\int_{t_i^n}^{t_{i+1}^n} \bar u^{t_i^n,y}(r)(X(r)+D^-_yJ(r,y)) dr |\mathcal{F}_t]
\\ && +  \sum_{i=0}^{n-1} E[ \int_{t_i^n}^{t_{i+1}^n}   \bar u^{t_i^n,y}(r) \left(\frac{J(t_{i+1},y+\bar Y_i^n)-J(t_{i+1},y)}{\bar Y_i^n}-D^-_yJ(r,y)\right)dr|\mathcal{F}_t] \\
&=& (I)+(II)
\end{eqnarray*}
 Then,
$$
(I)\leq E\left[\left. L \int_t^T (X(s)+D^-_yJ(s,y))_+ds\right|\mathcal{F}_t\right],
$$
and it remains to show that the limsup of  $(II)$ is nonpositive. We denote by $\widehat{D^-_yJ}(r,y)$ the RCLL modification of $D^-_yJ(r,y)$ which exists by Proposition \ref{prop:derivative}, (i).
Moreover, let $\bar \pi_n(r)=t_{i+1}^n$ for $r\in(t^n_i,t^n_{i+1}]$. By concavity
we get
\begin{eqnarray*}
 (II)&\leq& \sum_{i=0}^{n-1} E[ \int_{t_i^n}^{t_{i+1}^n}  \bar u^{t_i^n,y}(r) \left(D^-_yJ(t^n_{i+1},y)-D^-_yJ(r,y)\right)dr|\mathcal{F}_t] \\
&\leq& L E[\int_t^T |\widehat{D^-_yJ}(\bar \pi_n(r),y)-\widehat{D^-_yJ}(r,y)|dr |\mathcal{F}_t].
\end{eqnarray*}
The right-hand side converges to zero by rightcontinuity and dominated convergence. \\
{\it Step 2:} We show
$$
J(t,y)\geq E\left[\left. L \int_t^T (X(s)+D^-_yJ(s,y))_+ds\right|\mathcal{F}_t\right].
$$
for every
$$
y\in A:=\{\eta \in (-\infty,1);\ \lambda_{[0,T]}\otimes P(\{D^-_y J(\cdot,\eta)\neq D^+_y J(\cdot,\eta)\})=0\}.
$$
We fix a pair $(t,y)\in[0,T]\times A$ and choose a sequence of partitions $\pi_n=\{t^n_0,t^n_1,\ldots,t^n_n\}$ of $[t,T]$ such that the mesh size $|\pi_n|=\max_{i=1,\ldots, n} |t^n_i-t^n_{i-1}|$ tends to zero
and with $t^n_0=t$ and $t_n^n=T$. Now we define the controls,
$$
u_m^{t_i^n,y}(r)=L {\bf 1}_{(t_i^n, t_{i+1}^n]}(r) {\bf 1}_{\{Z_m(t_i^n)>0\}},\quad m\in \mathbb{N},
$$
where
$$
Z_m(r)=m \int_{(r-1/m)\wedge 0}^r (X(s)+D^-_yJ(s,y)) ds,\quad  r\in [0,T].
$$
By Lebesgue's differentiation theorem and Fubini's theorem
\begin{equation}\label{hilf15}
 \lambda_{[t,T]}\otimes P(\{(r,\omega);\; \lim_{m\rightarrow \infty} Z_m(r)= (X(r)+D^-_yJ(r,y)) \}^c)=0
\end{equation}
Note that  $u_m^{t_i^n,y} \in {U}(t_i^n,y)$ for sufficiently large $n$ (independent of $m$), which we assume from now on. We define
$$
Y_i^{n,m}:=\int_{t_i^n}^{t_{i+1}^n} u_m^{t_i^n,y}(r)dr,
$$
which is $\mathcal{F}_{t_i^n}$-measurable. Similarly to the first step, but taking the suboptimality of the controls into account, we get
\begin{eqnarray*}
  && J(t,y)\\&\geq& \sum_{i=0}^{n-1} E[\int_{t_i^n}^{t_{i+1}^n}  u_m^{t_i^n,y}(r)(X(r)+D^-_yJ(r,y)) dr |\mathcal{F}_t]
\\ && +  \sum_{i=0}^{n-1} E[ \int_{t_i^n}^{t_{i+1}^n}   u_m^{t_i^n,y}(r) \left(\frac{J(t_{i+1},y+ Y_i^{n,m})-J(t_{i+1},y)}{ Y_i^{n,m}}-D^-_yJ(r,y)\right)dr|\mathcal{F}_t] \\
&=& (I)+(II)
\end{eqnarray*}
We first treat the term $(I)$.  Let $\underline \pi_n(r)=t_{i}^n$ for $r\in(t^n_i,t^n_{i+1}]$. Then,
\begin{eqnarray*}
 (I)&=& E[\int_{t}^{T} L {\bf 1}_{\{Z_m(\underline{\pi}_n(r))>0\}} (X(r)+{D^-_yJ}(r,y)) dr |\mathcal{F}_t]
\\ &\geq& E[\int_{t}^{T} L {\bf 1}_{\{Z_m(\underline{\pi}_n(r))>0\}} Z_m(r) dr |\mathcal{F}_t]  \\ &&- L E[\int_{t}^{T}
|X(r)+{D^-_yJ}(r,y)-Z_m(r)|dr\;|\mathcal{F}_t]
\end{eqnarray*}
Concerning  term $(II)$ we note that, for $r\in (t^n_i,t^n_{i+1}]$,
\begin{eqnarray*}
 && E[\int_{t^n_i}^{t^n_{i+1}}  u_m^{t_i^n,y}(r)\left( \frac{J(t_{i+1},y+ Y_i^{n,m})-J(t_{i+1},y)}{ Y_i^{n,m}}-D^-_yJ(r,y)\right)dr|\mathcal{F}_t] \\
&=&  E[\int_{t^n_i}^{t^n_{i+1}}    u_m^{t_i^n,y}(r)\left(\frac{ E[J(t_{i+1},y+ Y_i^{n,m})-J(t_{i+1},y)|\mathcal{F}_r]}{ Y_i^{n,m}}-D^-_yJ(r,y)dr\right)|\mathcal{F}_t] \\
&\geq&  E[ \int_{t^n_i}^{t^n_{i+1}}   u_m^{t_i^n,y}(r)\left(\frac{ J(r,y+ Y_i^{n,m})-J(r,y)}{ Y_i^{n,m}}-D^-_yJ(r,y)\right)dr|\mathcal{F}_t]
\\ &=&  E[ \int_{t^n_i}^{t^n_{i+1}}   u_m^{t_i^n,y}(r)\left(\frac{ J(r,y+ L(t_{i+1}^n-t_i^n)-J(r,y)}{ L(t^n_{i+1} - t^n_i) }-D^-_yJ(r,y)\right)dr|\mathcal{F}_t] \\
&\geq& - L E[ \int_{t^n_i}^{t^n_{i+1}} |\frac{ J(r,y+ L(t_{i+1}^n-t_i^n)-J(r,y)}{ L(t^n_{i+1} - t^n_i) }-D^-_yJ(r,y)|\;dr\;|\mathcal{F}_t]
\end{eqnarray*}
Here, we applied the $\mathcal{F}_{t_i^n}$-measurability of $Y^{n,m}_i$ and the submartingale property in Proposition \ref{prop:submartingale}. Hence, making use of $y\in A$,
\begin{eqnarray*}
 && (II) \\
&\geq & -L E[ \int_t^T    |\frac{ J(r,y+ L(\bar \pi_n(r)-\underline \pi_n(r)))-J(r,y)}{ L(\bar \pi_n(r)-\underline \pi_n(r))}-D^+_yJ(r,y)|dr|\mathcal{F}_t].
\end{eqnarray*}
Gathering the estimates for $(I)$ and $(II)$ we have
\begin{eqnarray*}
 J(t,y)&\geq&   E[\int_{t}^{T} L {\bf 1}_{\{Z_m(\underline{\pi}_n(r))>0\}} Z_m(r) dr |\mathcal{F}_t]  \\ &&- L E[\int_{t}^{T}
|X(r)+{D^-_yJ}(r,y)-Z_m(r)|dr\;|\mathcal{F}_t]  \\ &&-L E[ \int_t^T    |\frac{ J(r,y+ L(\bar \pi_n(r)-\underline \pi_n(r)))-J(r,y)}{ L(\bar \pi_n(r)-\underline \pi_n(r))}-D^+_yJ(r,y)|dr|\mathcal{F}_t].
\end{eqnarray*}
As $Z_m$ has continuous paths, we
get $$L {\bf 1}_{\{Z_m(\underline{\pi}_n(r))>0\}} Z_m(r) dr\rightarrow L(Z_m(r))_+$$ as $n$ tends to infinity. Letting $n$ go to infinity, we thus obtain by dominated convergence
$$
J(t,y)\geq E[\int_{t}^{T} L (Z_m(r))_+ dr |\mathcal{F}_t] - L E[\int_{t}^{T}
|X(r)+{D^-_yJ}(r,y)-Z_m(r)|dr\;|\mathcal{F}_t].
$$
In view of (\ref{hilf15}) the proof of Step 2 can then be completed by letting $m$ tend to infinity.
\\
{\it Step 3:} We can now prove the assertion.

By step 1 and 2 we have
\begin{eqnarray}\label{eq:hilf9}
 J(t,y)&=&E\left[\left. L \int_t^T (X(s)+D^-_yJ(s,y))_+ds\right|\mathcal{F}_t\right]\nonumber \\
 &=&E\left[\left. L \int_t^T (X(s)+D^+_yJ(s,y))_+ds\right|\mathcal{F}_t\right]
\end{eqnarray}
for $y \in A$. Now fix some $y\in (-\infty,-1)\setminus A$. By Proposition \ref{prop:derivative}, (ii), there are sequences $(\bar y_k)$ and $(\underline y_k)$ in $A$ such that $\bar y_k \downarrow y$
and $\underline y_k \uparrow y$. Recalling that $y\mapsto J(t,y)$ is continuous, $y\mapsto D^-_yJ(s,y)$ is leftcontinuous and $y\mapsto D^+_yJ(s,y)$ is rightcontinuous, we immediately see
that the equations in (\ref{eq:hilf9}) also hold for $y$.
\end{proof}

We can slightly reformulate the result that the value process solves the above BSPDE in the following way.
\begin{cor}\label{cor:existence}
  For every $y\in (-\infty,1]$
\begin{eqnarray*}
 J(t,y)&=&E\left[\left. L \int_t^T (X(s)+D^-_yJ(s,y))_+ds\right|\mathcal{F}_t\right],  \quad t\in [0,T]\\
D^-_yJ(t,1)&\leq&-X(t), \quad D^-_y J(t,1-L(T-t))=0,  \quad t\in [0,T)
\end{eqnarray*}
holds $P$-almost surely.
\end{cor}
\begin{proof}
 In view of the previous theorem, we only need to show that
\begin{equation}\label{eq:hilf10}
D^-_yJ(t,1)\leq-X(t), \quad D^-_y J(t,1-L(T-t))=0,
\end{equation}
for every $t\in[0,T)$. The first assertion in (\ref{eq:hilf10}) in turn implies
$$
E\left[\left. L \int_t^T (X(s)+D^-_yJ(s,1))_+ds\right|\mathcal{F}_t\right]=0=J(t,1)
$$
for $t\in [0,T)$.

Note that the second assertion in (\ref{eq:hilf10}) is trivial, because
$$
J(t,y)=E[\int_t^T LX(s)ds|\mathcal{F}_t]
$$
for $y<1-L(T-t)$. In order to prove the first assertion we
define $u^{t,y}(r)=L{\bf 1}_{[t,t+(1-y)/L]}(r)$. Then, for $y<1$
\begin{eqnarray*}
 \frac{J(t,y)-J(t,1)}{y-1}=-\frac{J(t,y)}{1-y} \leq - \frac{L}{1-y} \int_t^{\min\{t+(1-y)/L,T\}} X(s)ds.
\end{eqnarray*}
By right-continuity of $X$, the right-hand side converges to $-X(t)$, which concludes the proof of (\ref{eq:hilf10}).
\end{proof}

\section{Characterization of optimal controls}

In this section we characterize optimality of controls. By Corollary \ref{cor:existence} one expects
that the following result holds under at most technical conditions: Suppose that $u\in U(t,y)$. Then $u$ is an optimal control, if and only if
\begin{equation}\label{eq:optimality}
 u(s) \in \left\{\begin{array}{cl} \{0\}, & X(s)+D_y^-J(s,y+\int_t^s  u(r)dr)<0\\ \{L\}, & X(s)+D_y^-J(s,y+\int_t^s u(r)dr)>0 \\
 \ [0,L], &
X(s)+D_y^-J(s,y+\int_t^s  u(r)dr)=0 \end{array}\right.
\end{equation}
$\lambda_{[t,T]}\otimes P$-almost surely.

To prove such result we require an appropriate version of a chain rule, which is derived in the following lemma.
\begin{lem}\label{lem:chainrule}
 Suppose
$$
V(t,y)=E[\int_{t}^T v(r,y) dr |\mathcal{F}_t],\quad t\in [0,T], y\in (-\infty,1],
$$
is an adapted random field which satisfies:
\begin{enumerate}
 \item There is a set $\bar \Omega$ of full $P$ measure such that $D^-_y V(t,\omega,y)$ exists  for every $t\in [0,T],\; y\in (-\infty,1]$ and $\omega \in \bar \Omega$, and such that
  $v(t,\omega,y)$ is leftcontinuous in $y$ for every $t\in [0,T],\; y\in (-\infty,1]$ and $\omega \in \bar \Omega$.
     \item $v(t,y)$ is $(\mathcal{F}_t)$-adapted for every $y\in(-\infty,1]$ and
 $$
E[\sup_{(t,y,\tilde y)\in [0,T]\times (-\infty,1]^2,\; \tilde y\neq y} \left(|v(t,y)|+ \left| \frac{V(t,y)-V(t,\tilde y)}{y-\tilde y}\right|\right)]<\infty.
$$
\end{enumerate}
Then, for every $(t,y)\in [0,T]\times (-\infty,1]$ and for every nondecreasing process
of the form $y(r)=y+\int_{t}^r u(s)ds$ with $u\in U(t,y)$,
$$
V(t,y)=E[\int_t^T  \left( v(r,y(r)) dr -  D^-_yV(r,y(r))u(r) \right)dr |\mathcal{F}_t]
$$
holds $P$-almost surely.
\end{lem}

\begin{proof}
 We first smoothen $V$ in $y$-direction by setting
$$
\tilde V(t,y):=\int_0^y V(t,\eta)d\eta,
$$
with the usual convention that $\int_0^y V(t,\eta) d\eta= - \int_y^0 V(t,\eta)d\eta$ for $y<0$.
Notice that
$$
\tilde V(t,y)=E[\int_t^T \tilde v(t,y) dt|\mathcal{F}_t]
$$
where
$$
\tilde v(t,y)=\int_0^y v(t,\eta)d\eta.
$$
We now fix a pair $(t,y)\in [0,T]\times (-\infty,1)$ and define, for $n\in \mathbb{N}$, $t_i^n:=t+i(T-t)/n$ and
\begin{eqnarray*}
U_n(t,y)&:=&\Bigl\{u\in U(t,y);\; u(r)=u(t_i^n), \textnormal{ for every } r\in[t_i^n,t_{i+1}^n) \Bigr\}
\end{eqnarray*}
{\it Step 1:} For $u\in U_n(t,y)$ and $y(r)=y+\int_t^r u(s)ds$
\begin{eqnarray*}
 \tilde V(t,y)=E[\int_t^T \tilde v(r,y(r))dr - V(r,y(r))u(r)dr|\mathcal{F}_t].
\end{eqnarray*}
In order to prove Step 1, we fix some $n\in \mathbb{N}$ and $u\in U_n(t,y)$. Choose a sequence of refining partitions $(\pi_N)_{N\geq n}$ of $[t,T]$ such that
$\{t^n_0,\ldots, t_n^n\}\subset \{s^N_0,\ldots, s_N^N \}=\pi_N$ for every $N\geq n$. We then define
$$
\underline\pi_N(r)=s_i^N,\quad \bar\pi_N(r)=s^N_{i+1},\quad r\in(s^N_i,s^N_{i+1}].
$$
We  split
\begin{eqnarray*}
 \tilde V(t,y)&=&E[\sum_{i=0}^{N-1} \tilde V(s^N_i, y(s^N_i))- \tilde V(s_{i+1}^N,y(s^N_{i+1}))|\mathcal{F}_t] \\
&=& E[\sum_{i=0}^{N-1} \tilde V(s^N_i, y(s^N_{i+1}))- \tilde V(s^N_{i+1},y(s^N_{i+1}))|\mathcal{F}_t] \\ && +E[\sum_{i=0}^{N-1} \tilde V(s^N_i, y(s^N_i))- \tilde V(s^N_i,y(s^N_{i+1}))|\mathcal{F}_t]
\\ &=& (I)+ ({II})
\end{eqnarray*}
Then,
\begin{eqnarray*}
  (I)&=& E[\int_{t}^{T} \tilde v(r,y(\bar \pi_N(r)))dr|\mathcal{F}_t].
\end{eqnarray*}
By continuity of $\tilde v(r,\cdot)$ and dominated convergence we obtain that
$$
(I)\rightarrow E[\int_{t}^{T} \tilde v(r,y(r))dr|\mathcal{F}_t],\quad N\rightarrow \infty.
$$
We now observe that
\begin{eqnarray*}
  ({II}) &=& E[\sum_{i=0}^{N-1} \int_{s_i^N}^{s^N_{i+1}} \frac{\tilde V(s^N_i, y(s^N_i))- \tilde V(s^N_i,y(s^N_{i+1}))}{s^N_{i+1}-s^N_{i}}dr|\mathcal{F}_t] \\&=&
E[\sum_{i=0}^{N-1} \int_{s_i^N}^{s^N_{i+1}} \frac{-1}{y(s^N_{i+1})-y(s^N_{i})} \int_{y(s^N_{i})}^{y(s^N_{i+1})} V(s^N_i,\eta)d\eta\; u(r)dr|\mathcal{F}_t].
\end{eqnarray*}
Here, we used that $u(r)=u(s_i^N)$ for $r\in [s^N_i,s_{i+1}^N)$ and  $y(s_{i+1}^N)=y(s_{i}^N)+ u(s_i^N)(s_{i+1}^N-s_i^N)$. Then,   $y(s_{i+1}^N)$ is $\mathcal{F}_{s_i^N}$-measurable
and, thus,
\begin{eqnarray*}
({II}) &=& E[\sum_{i=0}^{N-1} \int_{s_i^N}^{s^N_{i+1}} \frac{-1}{y(s^N_{i+1})-y(s^N_{i})} \int_{y(s^N_{i})}^{y(s^N_{i+1})} V(r,\eta) d\eta\; u(r)dr|\mathcal{F}_t] \\
 && +  E[\sum_{i=0}^{N-1} \int_{s_i^N}^{s^N_{i+1}} \frac{-1}{y(s^N_{i+1})-y(s^N_{i})} \\ &&\quad\quad\times \int_{y(s^N_{i})}^{y(s^N_{i+1})} E[V(s^N_i,\eta)-V(r,\eta)|\mathcal{F}_{s_{i}^N}] d\eta\; u(r)dr|\mathcal{F}_t] \\
&=:& (IIa)+(IIb).
\end{eqnarray*}
Then,
$$
(IIa)=E[ \int_{t}^{T} \frac{-1}{y(\bar\pi_N(r))-y(\underline\pi_N(r))} \int_{y(\underline\pi_N(r))}^{y(\bar\pi_N(r))} V(r,\eta) d\eta\; u(r)dr|\mathcal{F}_t]
$$
By continuity of $V(r,\cdot)$ and dominated convergence we get
$$
(IIa)\rightarrow -E[\int_{t}^{T} V(r,y(r))u(r)dr|\mathcal{F}_t], \quad N\rightarrow \infty.
$$
It thus remains to show that $(IIb)$ goes to zero. To see this we note that for $r\in [s^N_i,s^N_{i+1}]$
$$
 |E[V(s^N_i,\eta)-V(r,\eta)|\mathcal{F}_{s_{i}^N}] | \leq |\pi_N| E[\sup_{(s,\eta)} |v(s,\eta)||\mathcal{F}_{s_{i}^N}],
$$
where $|\pi_N|$ denotes the mesh size of the partition $\pi_N$.
Hence,
$$
|(IIb)|\leq    |\pi_N|  E[\int_{t}^{T} u(r)dr\sup_{(s,\eta)} |v(s,\eta)| |\mathcal{F}_t]  \rightarrow 0, \quad N\rightarrow \infty.
$$

{\it Step 2:} For $u\in U(t,y)$ and $y(r)=y+\int_t^r u(s)ds$
\begin{eqnarray*}
 \tilde V(t,y)=E[\int_t^T \tilde v(r,y(r))dr - V(r,y(r))u(r)dr|\mathcal{F}_t]
\end{eqnarray*}
 Indeed, given a control $u\in U(t,y)$, we define $u_n$ via
$$
u_n(r)= n\frac{ \int_{t_{i-1}}^{t_i} u(s)ds}{T-t}, \quad r\in [t_i^n,t_{i+1}^n),\;i=1,\ldots, n-1
$$
and $u_n(r)=0$ for $r\in[0,t_1^n)$.
Then, $u_n\in U_n(t,y)$.

Let $y(r)=y+\int_{t}^r u(s)ds$ and $y_n(r)=y+\int_{t}^r u_n(s)ds$. Then it is straightforward to verify that
$$
y(t_{i-1}^n) =y_n(t^n_i),\quad i=1,\ldots, n.
$$
This implies that the sequence $(y_n(r))$ converges to $y(r)$, as $n$ tends to infinity, for every $r\in[t,T]$. By
continuity of $\tilde v(r,\cdot)$ and $V(r,\cdot)$ and dominated convergence we have
\begin{eqnarray*}
 && E[\int_t^T \tilde v(r,y_n(r))dr - V(r,y_n(r))u(r)dr|\mathcal{F}_t]\\ &\rightarrow & E[\int_t^T \tilde v(r,y(r))dr - V(r,y(r))u(r)dr|\mathcal{F}_t].
\end{eqnarray*}
Moreover, $y_n(r) \rightarrow y(r)$ for every $r\in[t,T]$, together with the boundedness of the sequence $(u_n)$ in $L^2([t,T],\lambda_{[t,T]})$ implies that $(u_n)$ converges to $u$ weakly in
$L^2([t,T],\lambda_{[t,T]})$. Hence,
\begin{eqnarray*}
 E[\int_t^T  V(r,y(r))(u_n(r)-u(r)) dr |\mathcal{F}_t]\rightarrow 0 .
\end{eqnarray*}
This shows that Step 1 implies Step 2.

{\it Step 3:}  For $u\in U(t,y)$ and $y(r)=y+\int_t^r u(s)ds$
\begin{eqnarray*}
  V(t,y)=E[\int_t^T v(r,y(r))dr - D^-_yV(r,y(r))u(r)dr|\mathcal{F}_t].
\end{eqnarray*}
Fix some $u\in U(t,y)$ and note that $u$ also belongs to $U(t,y-\epsilon)$ for $\epsilon>0$. We apply Step 2 to get
\begin{eqnarray*}
 \tilde V(t,y)&=&E[\int_t^T \tilde v(r,y(r))dr - V(r,y(r))u(r)dr|\mathcal{F}_t]\\
\tilde V(t,y-\epsilon)&=&E[\int_t^T \tilde v(r,y(r)-\epsilon)dr - V(r,y(r)-\epsilon)u(r)dr|\mathcal{F}_t]
\end{eqnarray*}
where $y(r)=y+\int_t^r u(s)ds$. Hence,
\begin{eqnarray*}
 \frac{\tilde V(t,y-\epsilon)-\tilde V(t,y)}{-\epsilon}&=& E[\int_t^T \frac{1}{\epsilon} \int_{y(r)-\epsilon}^{y(r)} v(r,\eta)d\eta dr|\mathcal{F}_t] \\ &&
  - E[\int_t^T  u(r)\frac{V(r,y(r)-\epsilon)-V(r,y(r))}{-\epsilon} dr|\mathcal{F}_t]
\end{eqnarray*}
Letting $\epsilon$ tend to zero, the right-hand side converges to
$$
 E[\int_t^T v(r,y(r))-u(r) D^-_yV(r,y(r))dr|\mathcal{F}_t]
$$
by leftcontinuity of $v$, and the left-hand side converges to $V(t,y)$, because
$$
\frac{\tilde V(t,y-\epsilon)-\tilde V(t,y)}{-\epsilon}=  \frac{1}{\epsilon} \int_{y-\epsilon}^y V(t,\eta)d\eta.
$$

\end{proof}
By the results established in the previous sections
(Corollary \ref{cor:existence} and Proposition \ref{prop:goodversion}) we, hence, arrive at the following corollary.
\begin{cor}\label{cor:chainrulesolution}
 For every $(t,y)\in [0,T]\times (-\infty,1]$ and for every nondecreasing process
of the form $y(r)=y+\int_{t}^r u(s)ds$ with $u\in U(t,y)$,
$$
J(t,y)=E[\int_t^T  L(X(r)+ D^-_yJ(r,y(r)))_+dr - \int_t^T D^-_yJ(r,y(r))u(r)dr |\mathcal{F}_t]
$$
holds $P$-almost surely.
\end{cor}

We are now in the position to characterize the set of optimal controls.
\begin{thm}\label{thm:optimalcontrol}
 A control $u\in U(t,y)$ is optimal for $J(t,y)$, if and only if (\ref{eq:optimality}) holds.
\end{thm}
\begin{proof}
 By Corollary \ref{cor:chainrulesolution},
\begin{eqnarray*}
 && J(t,y)\\&=&E[\int_t^T  L(X(r)+ D^-_yJ(r,y(r)))_+dr - \int_t^T D^-_yJ(r,y(r))u(r)dr |\mathcal{F}_t] \\
 &=& E[\int_t^T X(r)u(r)dr |\mathcal{F}_t] \\ &&+
E[\int_t^T  \left( L(X(r)+ D^-_yJ(r,y(r)))_+ - (X(r) + D^-_yJ(r,y(r)))u(r) \right) dr |\mathcal{F}_t]
\end{eqnarray*}
Hence, $u$ is optimal, if and only if the nonnegative second term on
the right-hand side vanishes, which is equivalent to
(\ref{eq:optimality}).
\end{proof}

\section{A dual formulation}

We finally present a dual representation in terms of martingales. This type of representation was first suggested by \citet{Ro} and \citet{HK} for optimal stopping problems. A corresponding result
for general discrete time optimal control problems is due to \citet{BSS}.

The main idea is to relax the adaptedness condition on the set of
controls and to penalize non-adapted controls by a suitable choice
of martingales. We first introduce the set $\mathfrak{U}(t,y)$ of
deterministic functions $\mathfrak{u}:[t,T]\rightarrow [0,L]$ such
that   $\int_t^T \mathfrak{u}(s)ds\leq 1-y$. With this notation,
$U(t,y)$ is the set of adapted processes whose paths take values in
$\mathfrak{U}(t,y)$.
\begin{defi}
 Suppose $(t,y)\in [0,T]\times (-\infty,1]$. A map
$$
M:[t,T]\times \Omega \times \mathfrak{U}(t,y) \rightarrow \mathbb{R}
$$
is called a \emph{martingale map}, if $(M(s,u), \;t\leq s \leq T)$ is a martingale for every $u\in U(t,y)$. We denote the set of martingale maps by $\mathcal{M}(t,y)$.
\end{defi}
We can now represent the value $J(t,y)$ as a \iindex{solution of }
minimization problem over martingale maps.
\begin{thm}
  Suppose $(t,y)\in [0,T]\times (-\infty,1]$. Then,
\begin{equation} \label{eq:dual}
J(t,y)=\essinf_{M\in \mathcal{M}(t,y)} E[\esssup_{\mathfrak{u}\in \mathfrak{U}(t,y)} \int_t^T \mathfrak{u}(r)X(r) dr -(M(T,\mathfrak{u})-M(t,\mathfrak{u}))|\mathcal{F}_t].
\end{equation}
Moreover,
$$
\bar M^{t,y}(s,\mathfrak{u})=J(s,\mathfrak{y}(s))+\int_t^s L(X(r)+ D^-_yJ(r,\mathfrak{y}(r)))_+ - D^-_yJ(r,\mathfrak{y}(r))\mathfrak{u}(r)dr
$$
with $\mathfrak{y}(r)=y+\int_t^r \mathfrak{u}(l)dl$ is an optimal martingale map, which satisfies
\begin{equation}
J(t,y)=\esssup_{\mathfrak{u}\in \mathfrak{U}(t,y)} \int_t^T \mathfrak{u}(r)X(r) dr -(\bar M^{t,y}(T,\mathfrak{u})-\bar M^{t,y}(t,\mathfrak{u})).
\end{equation}
\end{thm}
\begin{proof}
Suppose $u\in U(t,y)$ and $M$ is a martingale map. Then,
\begin{eqnarray*}
 && E[\int_t^T {u}(r)X(r) dr|\mathcal{F}_t]=E[\int_t^T {u}(r)X(r) dr -(M(T,{u})-M(t,{u}))|\mathcal{F}_t] \\ &\leq &
 E[\esssup_{\mathfrak{u}\in \mathfrak{U}(t,y)} \int_t^T \mathfrak{u}(r)X(r) dr -(M(T,\mathfrak{u})-M(t,\mathfrak{u}))|\mathcal{F}_t],
\end{eqnarray*}
where the first identity is due to the martingale property of $M(s,u)$. This shows
\begin{equation*}
J(t,y)\leq \essinf_{M\in \mathcal{M}(t,y)} E[\esssup_{\mathfrak{u}\in \mathfrak{U}(t,y)} \int_t^T \mathfrak{u}(r)X(r) dr -(M(T,\mathfrak{u})-M(t,\mathfrak{u}))|\mathcal{F}_t].
\end{equation*}
In order to finish the proof it is now sufficient to show that $\bar M^{t,y}$ is a martingale map and satisfies
\begin{equation}\label{eq:hilf16}
J(t,y)\geq\esssup_{\mathfrak{u}\in \mathfrak{U}(t,y)} \int_t^T \mathfrak{u}(r)X(r) dr -(\bar M^{t,y}(T,\mathfrak{u})-\bar M^{t,y}(t,\mathfrak{u})).
\end{equation}
Fix some $u\in U(t,y)$ and let $y(r)=y+\int_t^r u(l)dl$. By Corollary \ref{cor:chainrulesolution}, we get for $s\geq t$
$$
J(s,y(s))=E[\int_s^T  L(X(r)+ D^-_yJ(r,y(r)))_+dr - \int_s^T D^-_yJ(r,y(r))u(r)dr |\mathcal{F}_s].
$$
This shows that $\bar M^{t,y}$ is a martingale map. Finally, (\ref{eq:hilf16}) holds, because
\begin{eqnarray*}
 &&\left( \esssup_{\mathfrak{u}\in \mathfrak{U}(t,y)} \int_t^T \mathfrak{u}(r)X(r) dr -(\bar M^{t,y}(T,\mathfrak{u})-\bar M^{t,y}(t,\mathfrak{u}))\right)-J(t,y)
\\ &=& \esssup_{\mathfrak{u}\in \mathfrak{U}(t,y)} \int_t^T \left(\mathfrak{u}(r)(X(r)+D^-_yJ(r,\mathfrak{y}(r)))  -L(X(r)+ D^-_yJ(r,\mathfrak{y}(r)))_+ \right)dr\\
&\leq & 0.
\end{eqnarray*}

\end{proof}

\subsubsection*{Acknowledgement}
The authors gratefully acknowledge financial support by the ATN-DAAD Australia Germany Joint Research Cooperation Scheme.

\end{document}